\documentclass[12pt]{amsart}
\usepackage[T1]{fontenc}
\usepackage{kpfonts}
\usepackage[dvipsnames]{xcolor}
\usepackage{pdfpages}
\usepackage{sgame}
\usepackage{accents}
\usepackage{tikz-cd}
\usepackage{float}
\usepackage[round]{natbib}
\usepackage{multirow}
\usepackage{dutchcal}
\usepackage{pdfpages}
\usepackage{bbm}
\usepackage{sgame}
\usepackage{mathtools}
\usepackage{amsthm}
\usepackage{enumitem}
\usepackage[margin=1in]{geometry}
\usepackage{caption}
\usepackage{subcaption}
\usepackage{epigraph}
\usepackage{listings}
\usetikzlibrary{calc}
\usetikzlibrary{shapes,arrows}
\usepackage{tcolorbox}

\newcommand{\ubar}[1]{\underaccent{\bar}{#1}}

\usepackage{hyperref,cleveref}
\newtheorem{theorem}{Theorem}[section]
\newtheorem{proposition}[theorem]{Proposition}
\newtheorem{lemma}[theorem]{Lemma}
\newtheorem{corollary}[theorem]{Corollary}
\newtheorem{claim}[theorem]{Claim}
\newtheorem{definition}[theorem]{Definition}

\newtheorem{assumption}[theorem]{Assumption}

\hypersetup{
    colorlinks=true,
    linkcolor=OrangeRed,
    filecolor=MidnightBlue,      
    urlcolor=MidnightBlue,
    citecolor=MidnightBlue,
}

\definecolor{backcolour}{rgb}{0.63, 0.79, 0.95}

\lstdefinestyle{mystyle}{
  backgroundcolor=\color{backcolour},
  basicstyle=\ttfamily\footnotesize,
  breakatwhitespace=false,         
  breaklines=true,                 
  captionpos=b,                    
  keepspaces=true,                 
  numbers=left,                    
  numbersep=5pt,                  
  showspaces=false,                
  showstringspaces=false,
  showtabs=false,                  
  tabsize=2
}

\providecommand{\keywords}[1]{\textbf{\textit{Keywords:}} #1}
\providecommand{\jel}[1]{\textbf{\textit{JEL Classifications:}} #1}


\begin{document}
\author{Marilyn Pease \and Mark Whitmeyer}
\thanks{MP: Kelley School of Business, Indiana University, \href{marpease@iu.edu}{marpease@iu.edu} \& MW: Arizona State University. \href{mailto:mark.whitmeyer@gmail.com}{mark.whitmeyer@gmail.com}. Pease thanks the John Rau Kelley School of Business Faculty Fellowship. We thank David Easley, Ravi Jagadeesan, Jo\~{a}o Ramos, Alex Teytelboym, and a conference audience at Durham University for their feedback.}

\title{On Risk Aversion in Auctions}
\begin{abstract} We provide a unifying way to analyze how risk aversion changes bidding in auctions by asking which bids become more attractive as bidders become more risk averse. In first-price auctions, under two payoff conditions--winning is never worse than the outside option, and winning with a low bid is preferable to winning only with a high bid--greater risk aversion makes high bids more appealing. In second-price auctions with a known outside option, bidding more increases risk exposure conditional on winning, so greater risk aversion favors lower bids. We show these bid-level forces translate into corresponding equilibrium comparative statics.
\end{abstract}
\maketitle
\\
\jel{D0; D8}

\section{Introduction}

Auctions are among the most widely used market institutions, and bidders' risk attitudes are central both for predicting behavior within auctions and for auction-design itself. The classical theory makes this point starkly: once bidders are risk-averse, familiar equivalence results break down and details like the pricing rule can have first-order effects. At the same time, the standard intuition for why risk aversion matters often relies on very specific modeling assumptions. As a result, several basic questions remain surprisingly unsettled at a general level: \emph{when} does greater risk aversion push bids \emph{up}, \emph{when} does it push bids \emph{down}, and \emph{why}?

A foundational result in the classic auction literature is that in symmetric private-values environments, risk aversion induces more aggressive bidding in first-price auctions, in the sense that equilibrium bids lie above the risk-neutral benchmark.\footnote{\cite{riley1981risk,maskin1984first,holt1980competitivebidding,matthews1987comparing} are classic treatments.} The usual intuition is an ``insurance'' motive: in a first-price auction, a higher bid trades off a lower surplus conditional on winning against a higher probability of winning, and a more risk-averse bidder is willing to sacrifice expected surplus to reduce payoff risk by increasing the chance of winning.\footnote{For a focused overview of risk aversion and auction design, see \cite{vasserman2021risk}. Experimental surveys include \cite{kagel1995auctionsurvey,kagellevin2016auctionsurvey}.}

This classic result, however, is robust to some extensions and not others.\footnote{Reserves under risk aversion (\cite{humatthewszou2010reserve,humatthewszou2019lowreserve,hu2011biddernumber}) and financial constraints (\cite{chegale1998standard}) yield the classic result. In other extensions such as heterogeneity and revenue comparisons with arbitrary types (\cite{chegale2006revenuecomparisons}), an uncertain number of bidders (\cite{mcafee_mcmillan1987stochastic,levinozdenoren2004uncertain}), selective entry with risk aversion (\cite{liluzhao2015selectiveentry}), hidden reserves (\cite{litan2017hiddenreserve}), or sequential versus simultaneous auctions with risk aversion (\cite{chakraborty2019simultaneous}), the result does not necessarily go through. In still other environments, like security-bid auctions (\cite{bajoori2024risk}), higher risk aversion leads to lower bids.} This paper proposes a unifying answer to the question of when and why risk aversion pushes bids up or down. To do this, we rely on a simple behavioral idea: treat ``bidding more'' or ``bidding less'' as choosing between actions that generate different distributions of payoffs, and compare these actions (bids) by how they are affected by bidders' risk aversion. In order to operationalize this idea, we rely on a concept from \cite{safety}: bid \(b\) is ``safer" than another, \(b'\), if the set of beliefs at which \(b\) is preferred to \(b'\) expands in a bidder's risk aversion. That is, whenever a bidder prefers \(b\) to \(b'\), she continues to prefer it if she is more risk averse.

The first payoff of this approach is conceptual. It turns the familiar ``insurance'' story into a sharp, easily checkable statement about the environment. In a first-price auction, bidding higher typically raises the probability of winning but reduces the payoff conditional on winning. The classic conclusion--more risk aversion increases bids--suggests that a higher bid is safer than a lower bid. We formalize and clarify this insurance story by giving economically transparent conditions under which this safety ranking holds, showing that these conditions are what drive the classic result in the first-price setting. Specifically, we require that \emph{winning must not be worse than losing} (relative to the bidder's outside option), and \emph{low bids are better winners} (winning with a low bid provides more surplus to the bidder). When these conditions are met, the insurance intuition becomes literally correct: greater risk aversion tilts choices toward higher bids because higher bids are safer.\footnote{More broadly, a large literature studies how benchmark conclusions respond when bidders differ in risk attitudes, face borrowing or liquidity constraints, decide whether to enter, or when the seller uses additional design instruments such as reserve prices and entry fees. These extensions matter both because they are empirically relevant and because they can complicate simple ``format rankings.''}

The second payoff is portability across formats and informational structures. The safety lens does not rely on the special ordinary-differential equation structure of independent-private-values models, and it accommodates settings in which bidders face additional uncertainty about the payoff from winning and may have type-dependent outside options. In the first-price auction, the paper establishes a clean monotone-comparative static in the symmetric increasing equilibrium class: when the prior distribution is exchangeable and under mild regularity conditions, a more risk-averse population bids weakly higher pointwise. The argument proceeds by isolating a single ``marginal tradeoff'' object--how a bidder values a marginal increase in win probability relative to the outside option--and showing that this object moves monotonically with risk aversion because the relevant local deviation is a safety improvement.

A central theme in more recent work is that bidders often face \textit{ex post} (ensuing) risk--uncertainty about realized value or downstream payoff consequences that is resolved after winning. In such environments, the insurance intuition can reverse: bidding more may increase exposure to a risky win payoff, making lower bids more attractive for more risk-averse bidders. This is the intuition behind ``precautionary bidding,'' and it motivates why one can observe risk-averse bidders shading downward in second-price and related auctions when the relevant uncertainty is about the payoff conditional on winning rather than about the event of winning itself.\footnote{The precautionary bidding channel is developed in \cite{eso2004precautionary}. Work on information and uncertainty resolution in timber and related environments includes \cite{athey2001information}. Related work on \textit{ex post} uncertainty includes \cite{luo_perrigne_vuong2018expost} and related papers. Experimental evidence on precautionary bidding is provided in \cite{kocher_pahlke_trautmann2015precautionary}.}

Our paper connects to this channel through the same logic used for first-price auctions. In particular, in a second-price (and more generally \(k\)-/uniform-price) environment with uncertainty in the realized value conditional on winning, bidding more increases the chance of being exposed to a risky payoff. In such cases, the safer action is the lower bid, and greater risk aversion leads bidders to shade bids downward. Our paper formalizes this ``precautionary bidding'' force and proves a corresponding monotone-comparative static for equilibrium bidding in a broad class of interdependent/common-value environments.

Our contribution here is not to re-solve each environment, but to supply a diagnostic criterion: once the environment is specified, the key question becomes whether raising one's bid reduces or increases the relevant payoff risk (relative to the outside option and conditional on the winning event). This lens clarifies why the first-price ``risk aversion raises bids'' force is robust in some extensions and fragile in others, and it separates cases where risk aversion primarily raises bids (insurance against losing) from cases where it primarily depresses bids (avoidance of risky wins).

In settings with interdependent or affiliated values, revenue comparisons are sensitive to the informational structure (e.g., linkage), and adding risk aversion introduces further complications: the classic risk-neutral ranking favoring second-price formats can collide with the classic private-values risk-aversion ranking favoring first-price formats. We do not aim to provide a universal revenue-ranking theorem under interdependence. Instead, we isolate a robust behavioral mechanism: when the risk is borne upon winning (relative to a known outside option), bidding more mainly increases exposure to that risk, depressing bids, even in rich value environments once one restricts attention to a natural equilibrium class (symmetric and increasing).\footnote{Foundational interdependence and affiliation results appear in \cite{milgromweber1982theory}. Early risk-aversion analyses in interdependent settings include \cite{matthews1987comparing,eso2005optimalcorrelated}.}

Risk aversion is also central in multi-unit settings (e.g., treasury or procurement auctions) where bidders may be exposed to both allocation and payoff risk conditional on winning. While much of the theoretical literature focuses on equilibrium characterization and efficiency/revenue properties under specific primitives, our uniform-price result emphasizes a portable force: uniform-price auctions, the pivotal tradeoff parallels the second-price tradeoff--a marginal deviation swaps a sure outside option for a risky win payoff at a price pinned down by an opponent's bid. Under the same ``known outside option / risky win payoff'' structure, we obtain a downward monotone comparative static in risk aversion in a symmetric increasing equilibrium class.\footnote{For divisible-good and treasury auction theory, see \cite{back_zender1993divisible}. For multi-unit auction theory, see \cite{ausubel_cramton2004auctioning}. For structural evidence on risk aversion in procurement/scaling settings, see \cite{lu2008bidding,bolotnyy2023scaling}.}

A parallel empirical literature documents systematic departures from risk-neutral bidding predictions and develops strategies to recover risk aversion from auction data. Laboratory evidence has long found robust overbidding relative to risk-neutral benchmarks in first-price settings, and field applications commonly estimate substantial curvature in bidders' utilities.\footnote{Classic experimental evidence includes \cite{cox1983risk,cox1988theory}; surveys include \cite{kagel1995auctionsurvey,kagellevin2016auctionsurvey}. On structural estimation and validation, see \cite{bajari_hortacsu2005structural}.} At the same time, risk aversion creates a basic identification challenge: bid distributions alone typically do not separately identify values and preferences without additional variation (across formats, entry decisions, or other strategic margins). Our results speak to these empirical themes by offering sign predictions tied to primitives--in particular, whether higher bids primarily insure against losing or primarily increase exposure to risky wins--which can help motivate which sources of variation are informative in a given application.\footnote{For risk-neutral nonparametric identification in first-price auctions, see \cite{guerrepperrignevuong2000optimal}. For nonparametric identification of risk aversion under exclusion restrictions, see \cite{guerrepperrignevuong2009identification}. For semiparametric estimation with risk aversion, see \cite{campo_guerre_perrigne_vuong2011semiparametric}. For identification in standard auction models and inference in English auctions, see \cite{athey_haile2002identification,haile_tamer2003inference}.}

Overall, the results offer a unified view of risk aversion in auctions: risk aversion pushes behavior toward safer bids, but whether ``safer'' means more or less depends on which side of the tradeoff carries the risk--foregone surplus from paying more, or exposure to a risky win payoff. Our framework clarifies the scope of the canonical first-price prediction, explains why second-/uniform-price environments can exhibit the opposite phenomenon, and provides a broad set of conditions that can guide both theory and applied work when choosing models and interpreting bid data.

\section{Auxiliary Results}

Here, we describe the basic building blocks for our auction analysis, taken from \citet{safety}. Consider a single-agent decision problem faced by a subjective-utility maximizer, in which uncertainty is parameterized by a random state \(\theta\). The state \(\theta\) is an element of some compact and metrizable state-space \(\Theta\), endowed with the Borel \(\sigma\)-algebra. We denote the set of all Borel probability measures on \(\Theta\) by \(\Delta \left(\Theta\right)\), and assume that the agent holds a subjective belief \(\mu \in \Delta \left(\Theta\right)\). 

The agent possesses a nonempty set of actions \(A \subseteq \mathbb{R}^{\Theta}\), where each action is a continuous function from the set of states, \(\Theta\), to the set of outcomes (monetary payoffs) \(\mathbb{R}\). Let \(a_\theta \in \mathbb{R}\) denote the payoff to action \(a\) in state \(\theta\). The agent has a utility function defined on the outcome space \(u \colon \mathbb{R} \to \mathbb{R}\), which we assume is continuous and strictly increasing. We say that she becomes more risk averse if her utility function is \(\hat{u}\) where \(\hat{u} = \phi \circ u\) for some strictly increasing and weakly concave \(\phi\).

For any two distinct actions \(a, b \in A\), we define the set \(P_{a, b}\left(a\right)\) to be the subset of the probability simplex on which the agent weakly prefers action \(a\) to \(b\); formally,
\[P_{a,b}\left(a\right) \coloneqq \left\{\mu \in \Delta \left(\Theta\right) \colon \mathbb{E}_{\mu}u\left(a_\theta\right) \geq \mathbb{E}_{\mu}u\left(b_\theta\right)\right\}\text{.}\]
When the utility function is \(\hat{u}\), we define the set \(\hat{P}_{a,b}\left(a\right)\) in the analogous manner.

Our key concept is the ``safety'' binary relation.
\begin{definition}
    Action \(a\) is \emph{safer} than action \(b\), \(a \succeq_S b\), if for any strictly increasing, concave \(\phi\), \(P_{a,b}\left(a\right) \subseteq \hat{P}_{a,b}\left(a\right)\).
\end{definition}
The main result of \citet{safety} characterizes the safety relation in terms of actions' state-dependent monetary payoffs. Toward stating that result, for any two distinct actions \(a\) and \(b\) that do not dominate each other, we let \(\mathcal{A}\) denote the set of states in which \(a\) strictly preferred to \(b\), \(\mathcal{B}\) denote the set of states in which \(b\) is strictly preferred to \(a\), and \(\mathcal{C}\) denote the set of states in which the DM is indifferent between \(a\) and \(b\): \[\label{char1}\tag{\(1\)}\mathcal{A} \coloneqq \left\{\theta \in \Theta \colon a_{\theta} \ > \ b_{\theta}\right\}, \quad \mathcal{B} \coloneqq \left\{\theta \in \Theta \colon a_{\theta} \ < \ b_{\theta}\right\}, \quad \mathcal{C} \coloneqq \left\{\theta \in \Theta \colon a_{\theta} \ = \ b_{\theta}\right\}\text{.}\]
Then,
\begin{theorem}[Theorem~3.1 in \citet{safety}]\label{thm:safety}Let \(a\) and \(b\) be two distinct actions that do not dominate each other. Then, action \(a\) is safer than action \(b\) if and only if for each \(\theta \in \mathcal{A}\) and \(\theta' \in \mathcal{B}\), \(b_{\theta'} \geq a_{\theta}\) and \(a_{\theta'} \geq b_{\theta}\).
\end{theorem}

Theorem \ref{thm:safety} tells us that \(a\) is safer than \(b\) if and only if the payoffs to \(a\) lie within the convex hull of the payoffs to \(b\). In other words, \(a\) is safer if and only if choosing \(a\) and being ``wrong" (because the state is actually \(\theta^\prime\)) is (weakly) not worse than choosing \(b\) and being ``wrong," while choosing \(a\) and being ``right" is not better than choosing \(b\) and being ``right."

\section{Bidding and Risk Aversion}\label{sec:bidra}

We begin our formal analysis by investigating the static auction of a single, indivisible good from the perspective of an individual bidder who is participating in the auction. At the point of bidding, there is a simple dichotomy of \textit{certain} objects and \textit{random} objects. These could be the bidder's value for the good upon receipt, her outside option, others' bids in the auction, whether she gets the good, how much she has to pay upon winning or losing the auction, the reserve price, etc. We summarize the random generator of these unknown objects with a single variable, the \textit{state} \(\theta \in \Theta\), which we assume to be compact and metrizable. We assume that the bidder is a subjective utility maximizer, whose belief is a Borel probability distribution over states \(\mu \in \Delta(\Theta)\). A bidder's decision is how much to bid. Importantly, in this section, our analysis is not equilibrium analysis, as \(\mu\) is held fixed as we change a bidder's risk aversion--\(\mu\) is endogenous in equilibrium and should vary as we vary \(u\). \S\ref{sec:fpa} and \S\ref{sec:spa} analyze equilibria.

In the auctions to be considered, if a bidder's bid is strictly higher than the maximum of the highest bid made by the others and the reserve price, denoted \(\gamma_{\theta}\), the allocation is known: the bidder gets the item. Likewise, the allocation is also known if her bid is strictly lower than \(\gamma_{\theta}\): the bidder does not get the item. Of course, at the point of bidding \textit{whether} her bid is maximal is typically random. The bidder's value for obtaining the good, \(v_\theta\), and what she gets upon losing the auction, her outside option \(s_\theta\), are both potentially random.

The following bit of accounting will prove especially useful. For two bids \(a > b\), we define
\[\label{exp1}\tag{\(2\)}\begin{split}
    \Theta_{both} &\coloneqq \left\{\theta \in \Theta \colon \ b > \gamma_{\theta}, \quad \text{or} \quad b = \gamma_{\theta} \text{ and get the good after either bid}\right\},\\
    \Theta_{a} &\coloneqq \left\{\theta \in \Theta \colon \ \begin{array}{cc}  a = \gamma_{\theta} \text{ and get the good only after } a, & a > \gamma_{\theta} > b,\\ \text{or } b = \gamma_{\theta} \text{ and get the good only after } a & \end{array}\right\},\\
    \Theta_{neither} &\coloneqq \left\{\theta \in \Theta \colon \ a < \gamma_{\theta}, \quad \text{or} \quad a = \gamma_{\theta} \text{ and don't get the good after either bid}\right\}\text{.}
\end{split}\]
Recall \eqref{char1}, that given two undominated actions (bids) \(a, b \in A\), \(\mathcal{A}\) is the set of states in which \(a\) strictly preferred to \(b\), \(\mathcal{B}\) is the set of states in which \(b\) is strictly preferred to \(a\), and \(\mathcal{C}\) is the set of states in which the DM is indifferent. Then,
\begin{lemma}
    In the \(k\)th-price (\(k \geq 1\)) auction, \(\mathcal{A} \subseteq  \Theta_{a}\).
\end{lemma}
\begin{proof}
    This is implied by the fact that the different bids do not affect who gets the good for states in \(\Theta_{both} \cup \Theta_{neither}\). The only difference from our buyer's perspective is the amount she pays. Moreover, all of these auctions have the feature that, conditional on whether the buyer receives the good, the buyer pays more if she bids more. Consequently, \(\Theta_{both}\) and \(\Theta_{neither}\) are subsets of either \(\mathcal{B}\) or \(\mathcal{C}\), the states in which the lower bid is strictly better or equal to the higher bid.
\end{proof}

On the other hand, whether \(\Theta_{both}\) or \(\Theta_{neither}\) are subsets of \(\mathcal{B}\) or \(\mathcal{C}\) does depend on the specific auction.

\subsection{First-Price Auctions}

In the first-price auction, if the bidder bids \(a\), then she obtains the real numbers
\[v_\theta - a\text{, if she \textbf{wins} the good;} \qquad \text{and} \qquad s_\theta, \text{ if she \textbf{does not win} the good.}\]

First let us categorize the sets \(\Theta_{both}\) and \(\Theta_{neither}\). As the only difference between bids \(a\) and \(b\) for states in \(\Theta_{both}\) is the amount the bidder pays and \(a > b\), it must be that \(\Theta_{both} \subseteq \mathcal{B}\). That is, if the state lies in \(\Theta_{both}\)--both bids win--it must be that the monetary reward from the lower bid is higher. Second, it must be that \(\Theta_{neither} \subseteq \mathcal{C}\), as the bidder merely gets her outside option for both bids (she loses the auction). If the state is such that both bids are losers, the bidder is indifferent between the two bids, as it is a first-price auction.

\begin{definition}
    We say that \emph{winning cannot hurt} if \(\inf_{\theta \in \Theta} \left(v_\theta - a\right) \geq \sup_{\theta \in \Theta}s_{\theta}\) and \emph{low bids are better winners} if \(\inf_{\theta' \in \mathcal{B}} \left(v_{\theta'} - b\right) \geq \sup_{\theta \in \mathcal{A}} \left(v_{\theta} - a\right)\).
\end{definition}
\begin{proposition}\label[proposition]{fpa}
    The larger bid, \(a\), is safer than the smaller bid, \(b\), \(a \succeq_S b\), if winning cannot hurt and low bids are better winners.
\end{proposition}
A rough sketch of the proof is as follows, with the formal proof deferred to Appendix \ref{fpaproof}. Given that winning cannot hurt, think about $\mathcal{A}$ and $\mathcal{B}$. $\mathcal{A}$ is the set of states where the bidder can only win with the higher bid, and $\mathcal{B}$ is the set of states where both bids win, but $b$ is cheaper. Then, directly from \Cref{thm:safety}, the two conditions equivalent to safety are 
\[a_{\theta^\prime} = v_{\theta^\prime} - a \geq s_\theta = b_\theta, \quad \text{and} \quad b_{\theta^\prime} = v_{\theta^\prime} - b \geq v_\theta - a = a_\theta \quad \text{for all } \theta, \theta^\prime\text{.}\]
Then, winning cannot hurt implies the first, and low bids are better winners implies the second.

The intuition behind this result is straightforward. Under the partition in \eqref{exp1}, the bids differ only on \(\Theta_a\), the pivotal
states where bidding \(a\) changes the allocation relative to bidding \(b\). In states
\(\Theta_{\text{both}}\), both bids win, and the only effect of bidding higher is to pay more, moving
the payoff from \(v_\theta-b\) to \(v_\theta-a\). In states \(\Theta_{\text{neither}}\), both bids lose and the
bidder obtains \(s_\theta\) regardless, so she is indifferent. Thus, the substantive tradeoff occurs
on \(\Theta_a\): raising the bid replaces the losing payoff \(s_\theta\) (under \(b\)) with the winning
payoff \(v_\theta-a\) (under \(a\)).

The condition \emph{winning cannot hurt} ensures that these marginal wins never generate new
downside outcomes, so bidding higher cannot make the bidder worse off in states where it flips the allocation.
The condition \emph{low bids are better winners} ensures that the states where the lower bid
already wins (the set \(B=\Theta_{\text{both}}\)) are at least as favorable, in net-payoff terms, as
the marginal states where only the higher bid wins (the set \(A \subseteq \Theta_a\)). Together, these conditions mean that moving from \(b\) to \(a\) primarily removes probability
mass from the (potentially low) outside-option outcome and “finances” this by trimming payoffs
in the best winning states (those in which \(b\) already wins). This compresses payoffs in the
sense of Theorem~\ref{thm:safety}, so the set of beliefs under which the higher bid is preferred expands with
risk aversion. In this precise sense, a higher first-price bid acts like insurance against the risky
event of losing.

A special setting, which generalizes the independent private values environment, is that in which the bidder's valuation and the payoff from losing the bid--the outside option--are both known: \(v_{\theta} = v \in \mathbb{R}\) and \(s_{\theta} = s \in \mathbb{R}\) for all \(\theta \in \Theta\). We call this the \textit{known-values environment}.

\begin{corollary}
    In the known-values environment, the larger bid is safer than the smaller bid, \(a \succeq_S b\).
\end{corollary}
\begin{proof}
    As \(a\) is not dominated, \(v - a > s\). But then, \(\inf_{\theta \in \Theta} \left(v_{\theta} - a\right) = v- a > s = \sup_{\theta \in \Theta}s_{\theta}\), and \(\inf_{\theta' \in \mathcal{B}} \left(v_{\theta'} - b\right) = v-b > v- a = \sup_{\theta \in \mathcal{A}} \left(v_{\theta} - a\right)\), so winning cannot hurt and low bids are better winners. \Cref{fpa} then yields the result. \end{proof}

Next, we change the environment slightly and ask what happens when $v$ remains unknown, but $s$ is known so that \(s_{\theta} = s\) for all \(\theta \in \Theta\). 
\begin{lemma}
    In the known-outside-option environment, the larger bid is safer than the smaller bid, \(a \succeq_S b\), only if \(\mathcal{B} \cap \Theta_{a} = \emptyset\).
\end{lemma}
\begin{proof}
    Suppose for the sake of contraposition that \(\mathcal{B} \cap \Theta_{a} \neq \emptyset\). Then for some \(\theta \in \mathcal{A}\) and \(\theta' \in \mathcal{B} \cap \Theta_{a}\), we have \(a_{\theta} = v_{\theta} - a > s = b_{\theta'}\), so \(a \not\succeq_S \ b\).
\end{proof}
We see that in the known-outside-option environment, for a high bid to become more attractive as a result of increased unknown aversion, it must be that if bidding higher alters the allocation outcome with respect to our bidder--i.e., if bidding high makes the difference for getting the object--such an outcome must always be positive. In a sense, this is an extremely strong requirement; in particular, it disciplines the possible values the bidder can have for the object to a significant degree. Another way to view this is that if the high bid is safer, winning the good cannot hurt the buyer.

\subsection{Second-Price Auctions}

The ``second-'' modifier is merely cosmetic. Our setting is sufficiently abstract so that the results hold for any \(k\)th-price auction for \(k \geq 2\) (where \(\gamma_\theta\) is the \(k\)-th highest price). Now, if the bidder bids \(a\) and gets the good she obtains \(v_{\theta} - \gamma_{\theta}\) and if she doesn't, she gets \(s_{\theta}\).

Again, we fix two bids \(a > b\) and utilize the partition specified in \eqref{exp1}. Now, however, \(\Theta_{both} \subseteq \mathcal{C}\), as \(a_{\theta} = v_{\theta} - \gamma_{\theta} = b_{\theta}\) for all \(\theta \in \Theta_{both}\). That is, for states in which both bids win the object, the bidder is indifferent between the two bids. Likewise, \(\Theta_{neither} \subseteq \mathcal{C}\), as \(a_{\theta} = s_{\theta} = b_{\theta}\) for all \(\theta \in \Theta_{a \leq \gamma_{\theta}}\). For states in which both bids lose the object, the bidder is also indifferent. It remains to compare \(v_{\theta} - \gamma_{\theta}\) and \(s_{\theta}\) for \(\theta \in \Theta_{a}\).

\begin{proposition}
    If the outside option is known, i.e., \(s_{\theta} = s \in \mathbb{R}\) for all \(\theta \in \Theta_{a}\), the smaller bid is safer than the larger bid, \(b \succeq_S a\).
\end{proposition}
\begin{proof}
    Let neither \(a\) nor \(b\) be dominated. Suppose first \(s_{\theta} = s \in \mathbb{R}\) for all \(\theta \in \Theta_{a}\). Then, for all \(\theta \in \mathcal{A}\) and \(\theta' \in \mathcal{B}\), \(a_{\theta} = v_{\theta} - \gamma_{\theta} > s = b_{\theta} = b_{\theta'} = s > v_{\theta'} - \gamma_{\theta} = a_{\theta'}\), as desired.\end{proof}

In contrast to the first-price auction, in which increased risk aversion favors \textit{higher} bids; in the second-price auction, under the natural restriction of a known outside option, increased risk aversion favors \textit{lower} bids. Intuitively, because we need only compare states in $\Theta_a$, the only reason to prefer $b$ is if winning with $a$ is worse than the outside option. Then, as the bidder becomes more risk averse, taking the sure outside option (which the lower bid delivers) is safer. 

\section{Equilibrium Analysis: First-price Auctions}\label{sec:fpa}

With the local, partial-equilibrium analysis of the previous section in hand, we now ask what changing all players' level of risk aversion does to \textit{equilibrium} bidding. When beliefs are over the others' equilibrium strategies and changing risk aversion changes best responses, how is equilibrium bidding affected?

Consider a first-price auction with \(n\ge 2\) bidders. We analyze the known-values environment, where types \(V_1,\dots,V_n\in\left[\ubar{v},\bar{v}\right]\) have an atomless \emph{exchangeable} joint distribution. Bidders have type-dependent outside options given by \(s\colon \left[\ubar{v},\bar{v}\right] \to \mathbb{R}\). Then, a bidder of type \(v\) who submits bid \(b\) obtains monetary payoff \(v-b\) if she wins and \(s(v)\) if she loses.

Pre-transformation, each bidder has a common utility function \(u\colon \mathbb{R} \to \mathbb{R}\), which we assume is strictly increasing, weakly concave, and \(C^1\). Post-transformation, let \(\hat u=\phi\circ u\), where \(\phi\) is strictly increasing, weakly concave, and \(C^1\). 

\begin{assumption}\label{stand1}
    We restrict attention to symmetric BNE in which bidders choose strictly-increasing bid functions \(\beta, \hat{\beta} \colon \left[\ubar{v},\bar{v}\right] \to \mathbb{R}\) that satisfy
\begin{enumerate}
    \item \(\beta, \hat{\beta}\) are continuous on \(\left[\ubar{v},\bar{v}\right]\) and \(C^1\) on \(\left(\ubar{v},\bar{v}\right)\);
    \item For each \(v\in\left(\ubar{v},\bar{v}\right)\), the optimal bid is an interior optimum; and
    \item \(\beta(v), \hat{\beta}(v) <v-s(v)\) for all \(v\in\left(\ubar{v},\bar{v}\right)\), i.e., no type chooses a weakly dominated bid.
\end{enumerate}
\end{assumption}
Take such a symmetric, strictly increasing bid function \(\beta\). We do the standard trick of noting that, as \(\beta\) is strictly increasing, any bid \(b\) can be written uniquely as \(b=\beta(t)\) for \(t=\beta^{-1}(b)\). It is, therefore, convenient to parametrize deviations by a ``report'' \(t\in\left[\ubar{v},\bar{v}\right]\), interpreted as bidding \(\beta(t)\). Because all bidders use the same strictly increasing \(\beta\), a type-\(v\) bidder who reports \(t\) wins if and only if \(T \equiv\max_{j\neq i}V_j\le t\). We do likewise for \(\hat{\beta}\) post-transformation. With the \(t\) notation in hand, we define the conditional win probability \(q(v,t) \coloneqq \mathbb P \left(T\le t \,\middle|\, V_i=v\right)\). By construction, \(q(v,t)\) depends only on the joint distribution of types and not on \(\beta\) or \(\hat{\beta}\). We assume that this function \(q\) is \(C^1\) in \(t\) on \(\left(\ubar{v},\bar{v}\right)\) and satisfies \(q(v,v)>0\) and \(\partial_t q(v,t)\vert_{t=v}>0\) for all \(v\in\left(\ubar{v},\bar{v}\right)\). We further assume that at the interior optimum the optimal report \(t=v\).

Let \(E_\uparrow(u)\) denote the set of such symmetric strictly increasing equilibria under \(u\), and define \(E_\uparrow(\hat u)\) analogously under \(\hat u\). Finally, we impose the common boundary normalization that all equilibria in \(E_\uparrow(u)\cup E_\uparrow(\hat u)\) satisfy
\[\hat\beta\left(\ubar{v}\right) = \beta\left(\ubar{v}\right)=\ubar{b}
\quad\text{for a fixed constant }\ubar{b}\in\mathbb R,
\]
so, in particular, any \(\beta\in E_\uparrow(u)\) and \(\hat\beta\in E_\uparrow(\hat u)\) satisfy \(\beta\left(\ubar{v}\right)=\hat\beta\left(\ubar{v}\right)\).

\begin{theorem}\label{thm:fpa_safety_equilibrium_cs}For any \(\beta\in E_\uparrow(u)\) and any \(\hat\beta\in E_\uparrow(\hat u)\), \(\hat\beta(v) \ge \beta(v)\) for all \(v\in\left[\ubar{v},\bar{v}\right]\).
\end{theorem}
\begin{proof}
    See Appendix \ref{a:fpa_eq}
\end{proof}

\Cref{thm:fpa_safety_equilibrium_cs} provides an equilibrium counterpart to the bid-level “insurance” mechanism from
\S\ref{sec:bidra}. In a symmetric strictly increasing equilibrium, a type-\(v\) bidder can be viewed as choosing
a report \(t\), interpreted as bidding \(\beta(t)\). Increasing \(t\) raises the win probability \(q(v,t)\)
but lowers the payoff conditional on winning from \(v-\beta(v)\) to \(v-\beta(t)\), while the losing
payoff remains the known outside option \(s(v)\). Thus, the local deviation trades off a marginal
increase in winning probability against a marginal reduction in win surplus.

The equilibrium first-order condition can be written as
\[
\beta'(v)=\Lambda(v) M_u\left(v-\beta(v);v\right),
\quad \text{where} \quad
\Lambda(v)\equiv \frac{\partial_t q(v,t)\rvert_{t=v}}{q(v,v)} > 0,\]
and where \[
M_u(x;v)\coloneqq \frac{u(x)-u(s(v))}{u'(x)}
\]
captures a type-\(v\) bidder's marginal willingness to trade off conditional win surplus \(x\) for a
marginal increase in win probability away from the outside option. Under a concave transformation
\(\hat u=\phi\circ u\), this marginal tradeoff increases, \(M_{\hat u}(x;v)\ge M_u(x;v)\); that is, as higher bids are safer, more risk-averse bidders value this margin more. Because \(\Lambda(v)\) is pinned down by the
(exchangeable) joint distribution of types, the strengthened marginal valuation under \(\hat u\),
together with the common boundary normalization, yields the pointwise ordering \(\hat\beta(v)\ge\beta(v)\)
in \Cref{thm:fpa_safety_equilibrium_cs}.

\section{Equilibrium Analysis: Second-price Auctions}\label{sec:spa}

Consider now a second-price (Vickrey) auction with \(n\ge 2\) bidders. We maintain the assumptions that types \(V_1,\dots,V_n\in\left[\ubar{v},\bar{v}\right]\) have an atomless \emph{exchangeable} joint distribution and that bidders have known type-dependent outside options \(s\colon \left[\ubar{v},\bar{v}\right]\to\mathbb{R}\). However, we now assume that there is valuation uncertainty at the point of bidding, formalized as follows: if bidder \(i\) with type \(v\) wins and pays price \(p\), her monetary payoff is \(W(v,V_{-i})-p\), where \(V_{-i}\coloneqq (V_j)_{j\neq i}\) and \(W(v,V_{-i})\) is an \(\mathbb{R}\)-valued random variable. If she loses, her monetary payoff is \(s(v)\).

Pre-transformation, each bidder has a common utility function \(u\colon \mathbb{R}\to\mathbb{R}\), which we assume is strictly increasing, weakly concave, and \(C^1\). Post-transformation, let \(\hat u\coloneqq \phi\circ u\), where \(\phi\) is strictly increasing, weakly concave, and \(C^1\).

We maintain Assumption~\ref{stand1} (with the third bullet modified appropriately). Likewise, as before, we parametrize deviations by a report \(t\in\left[\ubar{v},\bar{v}\right]\), interpreted as bidding \(\beta(t)\). Recall \(T\coloneqq \max_{j\neq i}V_j\), and 
\(
q(v,t)\coloneqq \mathbb{P} \left(T\le t\,\middle|\,V_i=v\right).
\) We assume that \(q\) is \(C^1\) in \(t\) on \(\left(\ubar{v},\bar{v}\right)\) and satisfies \(\partial_t q(v,t)\vert_{t=v}>0\) for all \(v\in\left(\ubar{v},\bar{v}\right)\). Equivalently, \(T\mid V_i=v\) admits a conditional density \(f_{T\mid v}(t)\coloneqq \partial_t q(v,t)\) that is continuous and strictly positive at \(t=v\). We further assume that at the interior optimum the optimal report \(t=v\).

Let \(E_\uparrow(u)\) denote the set of such symmetric strictly increasing equilibria under \(u\), and define \(E_\uparrow(\hat u)\) analogously under \(\hat u\).

We also assume that for every \(v\in\left(\ubar{v},\bar v\right)\), every \(\beta\in E_{\uparrow}\left(u\right)\), and every \(\hat\beta\in E_{\uparrow}\left(\hat u\right)\), the functions
\[
z\mapsto E\left[u\left(W\left(v,V_{-i}\right)-\beta\left(z\right)\right)\mid V_i=v,T=z\right]
\]
and
\[
z\mapsto E\left[\hat u\left(W\left(v,V_{-i}\right)-\hat\beta\left(z\right)\right)\mid V_i=v,T=z\right]
\]
are continuous at \(z=v\).
\begin{theorem}\label{thm:spa_safety_equilibrium_cs}
For any \(\beta\in E_\uparrow(u)\) and any \(\hat\beta\in E_\uparrow(\hat u)\), \(\hat\beta(v)\le \beta(v)\) for all \(v\in\left[\ubar{v},\bar{v}\right]\).
\end{theorem}

The direction reverses in the second-price environment because a bidder's bid primarily affects
\emph{whether} she is exposed to a risky win payoff, not the payment conditional on winning in
pivotal states. In a symmetric strictly increasing equilibrium, when a type-\(v\) bidder reports \(t\)
(sending bid \(\beta(t)\)), she wins if and only if \(T \le t\). If she wins when \(T=z\),
she pays \(\beta(z)\), which is pinned down by the opponent's type rather than by her own bid.

At an interior optimum, the equilibrium condition reduces to pivotal indifference: conditional on
\(\{V_i=v, T=v\}\), the bidder is indifferent between losing (yielding the sure payoff \(s(v)\)) and
winning at the pivotal price \(\beta(v)\) (yielding the risky payoff \(W(v,V_{-i})-\beta(v)\)):
\[
\mathbb{E}\!\left[u\left(W(v,V_{-i})-\beta(v)\right)\,\middle|\, V_i=v, T=v\right]=u(s(v)).
\]
This local comparison is precisely between a sure outside option and a risky win payoff. Since the
outside option is the safer alternative, under a more risk-averse utility \(\hat u=\phi\circ u\) the bidder
weakly prefers to avoid winning at the old pivotal price:
\[
\mathbb{E}\!\left[\hat u\left(W(v,V_{-i})-\beta(v)\right)\,\middle|\, V_i=v, T=v\right]\le \hat u(s(v)).
\]
To restore pivotal indifference under \(\hat u\), the price at which the bidder is willing to become
pivotal must fall; equivalently, her bid must decrease. This yields the pointwise ordering
\(\hat\beta(v)\le \beta(v)\) in \Cref{thm:spa_safety_equilibrium_cs}. In first-price auctions, bidding higher insures
against the risky event of losing; here, bidding higher mainly increases exposure to risk conditional
on winning, so greater risk aversion generates precautionary downward shading.

\subsection{Uniform-Price Auctions}

The main analysis of this section extends in a straightforward manner to multi-unit uniform-price auctions with \(n\ge 3\) unit-demand bidders and
\(K\in\{2,\dots,n-1\}\) identical units. The allocation rule assigns the \(K\) units to the
\(K\) highest bids, and the per-unit price is the highest losing bid (equivalently, the \((K+1)\)st highest bid). We maintain the assumptions from this section, then note that the identical proof as for Theorem \ref{thm:spa_safety_equilibrium_cs} goes through, once we replace the highest opponent type \(T\) with $T^{(K)}$, the \(K\)-th highest opponent type--the \(K\)-th order statistic
of $(V_j)_{j\ne i}$ in descending order--and define
\[
q^{(K)}(v,t)\coloneqq \mathbb P(T^{(K)}\le t\mid V_i=v).
\]
\begin{corollary}
    For any $\beta\in E_\uparrow(u)$ and any $\hat\beta\in E_\uparrow(\hat u)$, $\hat\beta(v)\le \beta(v)$ for all \(v\in\left[\ubar{v},\bar v\right]\).
\end{corollary}

\appendix

\section{Omitted Proofs}

\subsection{\Cref{fpa} Proof}\label{fpaproof}
\begin{proof}
     As \(b\) does not dominate \(a\), \(\mathcal{A} \neq \emptyset\), which implies that there exists \(\theta \in \Theta_{a}\) such that \(v_\theta - a > s_\theta\). By our previous discussion, \(\mathcal{A} \subseteq \Theta_{a}\). We see that for all \(\theta \in \mathcal{A}\), we have
    \[a_\theta = v_\theta - a > s_\theta = b_\theta\text{,}\]
    and for all \(\theta' \in \mathcal{B}\), \(a_{\theta'} =  v_{\theta'} - a\). Finally,
    \[b_{\theta'} = \begin{cases}
    v_{\theta'} - b \quad &\text{for all} \quad \theta' \in \Theta_{both}\\
        s_{\theta'} \quad &\text{for all} \quad \theta' \in \mathcal{B} \cap \Theta_{a}\text{.}
    \end{cases}\]

    Therefore, \(a \succeq_S b\) if and only if 
    \[\label{in1}\tag{\(A1\)}\inf_{\theta' \in \mathcal{B}} \left(v_{\theta'} - a\right) \geq \sup_{\theta \in \mathcal{A}} s_\theta\text{,}\]
    and
    \[\label{in2}\tag{\(A2\)}\min\left\{\inf_{\theta' \in \Theta_{both}} \left(v_{\theta'} - b\right) , \inf_{\theta' \in \mathcal{B} \cap \Theta_{a}} s_{\theta'} \right\} \geq \sup_{\theta \in \mathcal{A}} \left(v_{\theta} - a\right)\text{.}\]

    Now suppose winning cannot hurt. Thus, as \(\inf_{\theta \in \Theta} \left(v_\theta - a\right) \geq \sup_{\theta \in \Theta}s_{\theta}\), \eqref{in1} must hold. Moreover, it must also be the case that \(\mathcal{B} \cap \Theta_{a} = \emptyset\), i.e., \(\mathcal{B} = \Theta_{b\geq\gamma_{\theta} }\). Thus, \eqref{in2} simplifies to
    \[\inf_{\theta' \in \mathcal{B}} \left(v_{\theta'} - b\right) \geq \sup_{\theta \in \mathcal{A}} \left(v_{\theta} - a\right)\text{,}\]
    which is low bids are better winners. 
\end{proof}

\subsection{\Cref{thm:fpa_safety_equilibrium_cs} Proof}\label{a:fpa_eq}

\begin{proof}
The type-\(v\) expected utility from reporting \(t\) under utility \(u\) is
\[
\Psi_u(v,t;\beta)
\coloneqq
q(v,t)u\left(v-\beta(t)\right)
+
\left(1-q(v,t)\right)u\left(s(v)\right).
\]
If \(\beta\in E_\uparrow(u)\), then for each \(v\), the report \(t=v\) maximizes \(\Psi_u(v,t;\beta)\).

Next, take \(v\in\left(\ubar{v},\bar{v}\right)\). By the interiority assumption, the maximizer is interior at \(t=v\), so
\[
\partial_t \Psi_u(v,t;\beta)\big\vert_{t=v}=0.
\]
We differentiate \(\Psi_u\) with respect to \(t\):
\[
\partial_t \Psi_u(v,t;\beta)
=
\partial_t q(v,t)\left(u\left(v-\beta(t)\right)-u\left(s(v)\right)\right)
-
q(v,t) u' \left(v-\beta(t)\right) \beta'(t).
\]
Evaluating at \(t=v\) and writing \(x(v) \coloneqq v-\beta(v)\) yields
\[
\beta'(v)
=
\Lambda(v) M_u\left(x(v);v\right),
\quad \text{where} \quad
\Lambda(v) \coloneqq \frac{\partial_t q(v,t)\vert_{t=v}}{q(v,v)}  > 0,
\quad \text{and}
 \quad M_u(x;v) \coloneqq \frac{u(x)-u\left(s(v)\right)}{u'(x)}.
\]
The analogous relation holds for any \(\hat\beta\in E_\uparrow(\hat u)\). Writing \(\hat x(v) \coloneqq v-\hat\beta(v)\),
\[
\hat\beta'(v)
=
\Lambda(v) M_{\hat u} \left(\hat x(v);v\right),
\quad \text{where} \quad
M_{\hat u}(x;v) \coloneqq \frac{\hat u(x)-\hat u\left(s(v)\right)}{\hat u'(x)}.
\]

\begin{claim}
    \(M_{\hat u}(x;v)\ge M_u(x;v)\).
\end{claim}
\begin{proof}
    Take arbitrary \(v\in\left(\ubar{v},\bar{v}\right)\) and  \(x>s(v)\), and choose \(\delta\in(0,x-s(v))\). Consider two bids \(b<b+\delta\) that yield win payoffs \(x=v-b\) and \(x-\delta=v-(b+\delta)\), and losing payoff \(s(v)\). As \(b+\delta\) is safer than \(b\), we have
    \[
\frac{u(x)-u(x-\delta)}{u(x-\delta)-u\left(s(v)\right)}
 \geq 
\frac{\hat u(x)-\hat u(x-\delta)}{\hat u(x-\delta)-\hat u\left(s(v)\right)}.
\]
Dividing both sides by \(\delta>0\), we obtain
\[
\frac{\frac{u\left(x\right)-u\left(x-\delta\right)}{\delta}}{u\left(x-\delta\right)-u\left(s(v)\right)}\ge\frac{\frac{\hat u\left(x\right)-\hat u\left(x-\delta\right)}{\delta}}{\hat u\left(x-\delta\right)-\hat u\left(s(v)\right)}.
\]
Now let \(\delta\downarrow0\). By differentiability of \(u\) and \(\hat u\) at \(x\),
\[
\frac{u\left(x\right)-u\left(x-\delta\right)}{\delta}\to u'\left(x\right),\qquad \text{and} \qquad\frac{\hat u\left(x\right)-\hat u\left(x-\delta\right)}{\delta}\to\hat u'\left(x\right),
\]
and by continuity,
\[
u\left(x-\delta\right)-u\left(s(v)\right)\to u\left(x\right)-u\left(s(v)\right)>0,\quad \text{and} \quad\hat u\left(x-\delta\right)-\hat u\left(s(v)\right)\to\hat u\left(x\right)-\hat u\left(s(v)\right)>0.
\]
Therefore,
\[
\frac{u'\left(x\right)}{u\left(x\right)-u\left(s(v)\right)}\ge\frac{\hat u'\left(x\right)}{\hat u\left(x\right)-\hat u\left(s(v)\right)},
\]
as desired.
\end{proof}

Next we argue that for a fixed \(v\), the function \(x\mapsto M_{\hat u}(x;v)\) is strictly increasing on \(x>s(v)\). Write \(s \equiv  s(v)\), and take \(y>x>s\). Then
\[
M_{\hat u}(y;v)-M_{\hat u}(x;v)
=
\frac{\hat u(y)-\hat u(s)}{\hat u'(y)}
-
\frac{\hat u(x)-\hat u(s)}{\hat u'(x)}.
\]
Add and subtract \(\hat u(x)/\hat u'(y)\) to obtain
\[
M_{\hat u}(y;v)-M_{\hat u}(x;v)
=
\left(\hat u(x)-\hat u(s)\right)\left(\frac{1}{\hat u'(y)}-\frac{1}{\hat u'(x)}\right)
+
\frac{\hat u(y)-\hat u(x)}{\hat u'(y)}.
\]
Since \(\hat u\) is concave, \(\hat u'\) is nonincreasing, so \(\hat u'(y)\le \hat u'(x)\) and the first term is weakly positive.
Since \(\hat u\) is strictly increasing, \(\hat u(y)-\hat u(x)>0\) and \(\hat u'(y)>0\), so the second term is strictly positive.
Hence, \(M_{\hat u}(y;v)-M_{\hat u}(x;v)>0\).

Now fix \(\beta\in E_\uparrow(u)\) and \(\hat\beta\in E_\uparrow(\hat u)\), and define \(d(v) \coloneqq \hat\beta(v)-\beta(v)\).
\begin{claim}\label{eq:key_imp}
    For every \(v\in\left(\ubar{v},\bar{v}\right)\), if \(d(v)<0\) then \(d'(v)>0\).
\end{claim}
\begin{proof}
    If \(d(v)<0\) then \(\hat x(v)=v-\hat\beta(v)>v-\beta(v)=x(v)\). Consequently, \(M_{\hat u}(\hat x(v);v)>M_{\hat u}(x(v);v)\ge M_u(x(v);v)\). Since \(\Lambda(v)>0\), the FOCs imply \(d'(v)=\Lambda(v)\left(M_{\hat u}(\hat x(v);v)-M_u(x(v);v)\right)>0\).\end{proof}

Now suppose for the sake of contradiction that \(d(v_0)<0\) for some \(v_0\in\left(\ubar{v},\bar{v}\right)\).
Let
\[
\mathcal N \coloneqq \{v\in\left(\ubar{v},\bar{v}\right)\colon d(v)<0\}.
\]
Then \(\mathcal N\) is open, hence, a union of disjoint open intervals. Let \((a,b)\) be any connected component of \(\mathcal N\).
By the continuity of \(d\), we have \(d(a)=0\) if \(a>\ubar{v}\), and by the boundary normalization \(d(\ubar{v})=0\) we also have \(d(a)=0\) if \(a=\ubar{v}\).
Moreover, for every \(v\in(a,b)\) we have \(d(v)<0\), so by Claim \ref{eq:key_imp} we have \(d'(v)>0\) on \((a,b)\).
By the mean value theorem, this implies \(d\) is strictly increasing on \((a,b)\). Therefore, for any \(v\in(a,b)\),
\[
d(v)\ >\ \lim_{t\downarrow a} d(t)\ =\ d(a)\ =\ 0,
\]
contradicting that \(d(v)<0\) on \((a,b)\). Hence, \(\mathcal N\) must be empty, i.e. \(d(v)\ge 0\) for all \(v\in\left(\ubar{v},\bar{v}\right)\).
By continuity, \(d(v)\ge 0\) on \(\left[\ubar{v},\bar{v}\right]\); equivalently, \(\hat\beta(v)\ge \beta(v)\) for all \(v\in\left[\ubar{v},\bar{v}\right]\).
\end{proof}

\subsection{\Cref{thm:spa_safety_equilibrium_cs} Proof}

\begin{proof}
Because all bidders use the same strictly increasing \(\beta\), a type-\(v\) bidder who reports \(t\) wins if and only if \(T\le t\). In a second-price auction, if she wins when \(T=z\), she pays the highest losing bid, which equals \(\beta(z)\). Therefore, conditional on \(V_i=v\), her expected utility from reporting \(t\) is
\[
\Psi_u(v,t;\beta)
\coloneqq
\mathbb{E} \left[
u \left(W(v,V_{-i})-\beta(T)\right)\mathbf{1}\{T\le t\}
+
u \left(s(v)\right)\mathbf{1}\{T>t\}
\ \middle|\ V_i=v
\right].
\]
If \(\beta\in E_\uparrow(u)\), then for each \(v\), the report \(t=v\) maximizes \(\Psi_u(v,t;\beta)\).

Fix \(v\in\left(\ubar{v},\bar{v}\right)\). By our interiority assumption, the maximizer is interior at \(t=v\), so
\[
\partial_t \Psi_u(v,t;\beta)\big\vert_{t=v}=0.
\]
Define
\[
m_u(v,z)\coloneqq \mathbb{E} \left[u \left(W(v,V_{-i})-\beta(z)\right)\,\middle|\,V_i=v,\ T=z\right].
\]
Conditioning on \(T\) and using \(\mathbf{1}\{T\le t\}\), we can write
\[
\Psi_u(v,t;\beta)
=
\int_{\ubar v}^{t} m_u(v,z) f_{T\mid v}(z) dz
+
u \left(s(v)\right)\int_{t}^{\bar v} f_{T\mid v}(z) dz.
\]

By assumption \(z\mapsto m_u\left(v,z\right)f_{T|v}\left(z\right)\) is continuous at \(z=v\). Therefore,
\[
\left.\partial_t\int_{\ubar{v}}^{t}m_u\left(v,z\right)f_{T|v}\left(z\right) dz\right|_{t=v}=m_u\left(v,v\right)f_{T|v}\left(v\right)
\]
and
\[
\left.\partial_t\int_t^{\bar v}f_{T|v}\left(z\right) dz\right|_{t=v}=-f_{T|v}\left(v\right).
\]
Hence,
\[
\left.\partial_t\Psi_u\left(v,t;\beta\right)\right|_{t=v}=f_{T|v}\left(v\right)\left(m_u\left(v,v\right)-u\left(s(v)\right)\right).
\]
Since \(f_{T|v}\left(v\right)>0\), the first-order condition at \(t=v\) implies the pivotal indifference condition
\[
E\left[u\left(W\left(v,V_{-i}\right)-\beta\left(v\right)\right)\mid V_i=v,T=v\right]=u\left(s(v)\right).
\]

The same argument applied to any \(\hat\beta\in E_\uparrow(\hat u)\) yields
\[
\mathbb{E} \left[\hat u \left(W(v,V_{-i})-\hat\beta(v)\right)\,\middle|\,V_i=v,\ T=v\right]
=
\hat u(s(v)).
\tag{A3}\label{eq:spa_pivotal_uhat}
\]

Fix \(v\in\left(\ubar{v},\bar{v}\right)\) and set \(b\coloneqq \beta(v)\). Conditional on the pivotal event \(\{V_i=v, T=v\}\), the bidder is indifferent under \(u\) between the lower bid that loses (yielding the sure payoff \(s(v)\)) and the higher bid that wins at price \(b\) (yielding the random payoff \(W(v,V_{-i})-b\)). As the lower bid is safer, the bidder must weakly prefer it under \(\hat u=\phi\circ u\), i.e.
\[
\mathbb{E} \left[\hat u \left(W(v,V_{-i})-\beta(v)\right)\,\middle|\,V_i=v,\ T=v\right]
\ \le\
\hat u(s(v)).
\tag{A4}\label{eq:spa_one_sided}
\]

Fix \(v\in\left(\ubar{v},\bar{v}\right)\) and define
\[
\rho_v(b)\coloneqq \mathbb{E} \left[\hat u \left(W(v,V_{-i})-b\right)\,\middle|\,V_i=v,\ T=v\right].
\]
\(\hat u\) is strictly increasing, so \(\rho_v\) is strictly decreasing in \(b\). \eqref{eq:spa_one_sided} implies \(\rho_v(\beta(v))\le \hat u(s(v))\), and \eqref{eq:spa_pivotal_uhat} implies \(\rho_v(\hat\beta(v))=\hat u(s(v))\). \(\rho_v\) is strictly decreasing, so \(\hat\beta(v)\le \beta(v)\). As \(v\) was arbitrary, this holds for all \(v\in\left(\ubar{v},\bar{v}\right)\), and by continuity of \(\beta,\hat\beta\) also at the endpoints.
\end{proof}

\bibliography{sample.bib}
\bibliographystyle{plainnat}

\end{document}